\documentclass[10pt,journal]{IEEEtran}

\usepackage[utf8]{inputenc}
\usepackage[T1]{fontenc}
\usepackage[textwidth=16cm,textheight=21.7cm]{geometry}

\usepackage{graphicx}
\usepackage{physics}
\usepackage[nolist,nohyperlinks]{acronym}
\usepackage{amsmath}
\usepackage{amsthm}
\usepackage{amssymb}
\usepackage{authblk}

\usepackage{bbm}
\usepackage{subcaption}
\usepackage{cite}

\usepackage[ruled,vlined,linesnumbered]{algorithm2e}

\usepackage{tikz}
\usetikzlibrary{automata, arrows, positioning}

\usepackage{tikz-network}

\newcommand{\hilb}[0]{%
  \mathcal{H}%
}

\newcommand{\hilbdm}[0]{%
  \mathcal{D}%
}

\newcommand{\thetahat}[0]{%
  \hat{\theta}%
}

\newcommand{\ch}[0]{%
  \mathcal{E}%
}

\newtheorem{problem}{Problem}
\newtheorem{theorem}{Theorem}

\usepackage{acronym}
\newacro{CNOT}{controlled $X$}
\newacro{GHZ}{Greenberger-Horne-Zeillinger}
\newacro{QFI}{Quantum Fisher Information}
\newacro{QFIM}{Quantum Fisher Information Matrix}

\title{Quantum Network Tomography with Multi-party State Distribution}
\author[1]{Matheus Guedes de Andrade}
\author[2]{Jaime Días}
\author[2]{Jake Navas}
\author[3]{Saikat Guha}
\author[2]{Inès Montaño}
\author[4]{Brian Smith}
\author[4]{Michael Raymer}
\author[1]{Don Towsley}

\affil[1]{College of Information and Computer Science, University of Massachusetts Amherst}
\affil[3]{College of Optical Sciences, University of Arizona}\date{January 2022}
\affil[2]{Department of Applied Physics and Materials Science, Northern Arizona University}
\affil[4]{Department of Physics, University of Oregon}

\begin{document}

\maketitle

\begin{abstract}
    The fragile nature of quantum information makes it practically impossible to completely isolate a quantum state from noise under quantum channel transmissions. Quantum networks are complex systems formed by the interconnection of quantum processing devices through quantum channels. In this context, characterizing how channels introduce noise in transmitted quantum states is of paramount importance. Precise descriptions of the error distributions introduced by non-unitary quantum channels can inform quantum error correction protocols to tailor operations for the particular error model. In addition, characterizing such errors by monitoring the network with end-to-end measurements enables end-nodes to infer the status of network links.
    In this work, we address the end-to-end characterization of quantum channels in a quantum network by introducing the problem of Quantum Network Tomography. The solution for this problem is an estimator for the probabilities that define a Kraus decomposition for all quantum channels in the network, using measurements performed exclusively in the end-nodes. We study this problem in detail for the case of arbitrary star quantum networks with quantum channels described by a single Pauli operator, like bit-flip quantum channels. We provide solutions for such networks with polynomial sample complexity. Our solutions provide evidence that pre-shared entanglement brings advantages for estimation in terms of the identifiability of parameters.
\end{abstract}

\section{Introduction}

Quantum networks are communication systems formed by the interconnection of quantum processors with channels that enable the communication of quantum information~\cite{kimble2008quantum, van2014quantum, wehner2018quantum, wojciech2019networks}.
They extend the capabilities of quantum computers, augmenting the computing power of interconnected quantum processors with distributed quantum computing~\cite{beals2013efficient, van2016local}, and enable novel applications such as quantum key distribution~\cite{bennett2020quantum}.
As with any quantum processing system, noise is inherent to quantum networks due to the fragile nature of quantum information.

In networking systems, noise arises both in the communication of quantum information through channels and in the quantum operations performed by the nodes. The diverse ecosystem of physical platforms for qubit implementations, as well as the communication media used, have drastic impacts on errors introduced during communication, such that the development of technologies to transduce quantum information is critical~\cite{rakher2010quantum ,PRXQuantum.2.017002}. With respect to communication, photons acquire different errors when propagated through optical fiber and when propagated through free-space~\cite{peng2005experimental}. In the scenario of heterogeneous quantum networks that, for instance, make use of satellites and optical fibers to enable quantum communication, characterizing communication errors is central.
For instance, quantum error correction processs~\cite{grassl2018quantum} and entanglement purification methods~\cite{rozpkedek2018optimizing} can benefit from a precise description of error models to improve efficiency.

In the context of characterizing quantum systems, the theory of quantum estimation determines the fundamental limits to what can be inferred from measurements and the efficiency of realizable estimators~\cite{helstrom1969quantum, liu2019quantum}. Thus, quantum estimation is the key tool to describe errors introduced by noise in communication in a quantum network. In quantum information, Quantum Tomography refers to a set of methods for inferring the description of quantum systems~\cite{d2003quantum}.
Quantum State Tomography~\cite{smithey1993measurement} and Quantum Process Tomography~\cite{mohseni2008quantum} refer to methods that aim to provide descriptions of quantum states and of quantum evolution processes, respectively. Hence, the natural approach to characterize errors in a quantum network is to use quantum tomography methods. In this case, the network would be treated as a black-box and probed through the transmission of quantum states which would be measured to estimate the behavior of the entire network.

Unfortunately, such methods are not directly applicable if one wants to characterize the behavior of individual network channels by using the network to transmit quantum information. Consequently, investigating methods that use prior knowledge on the topology of the network, as well as on the form of the errors that arise from noise introduced by quantum channel transmissions brings a new perspective to such a tomography problem.

Exploiting prior knowledge to characterize channels in classical networks has been previously investigated and motivated the development of classical Network Tomography~\cite{castro2004network, duffield2006network}. Classical Network Tomography refers to a set of methods that aim to estimate the status of network links through end-to-end measurements. The key idea is to transmit messages among users in the network with the necessary meta-information to measure classical quantities such as transmission delay and packet loss. In the classical network scenario, tomography methods that use both unicast and multicast communication were proposed. Interestingly, link correlations introduced by multicast communication in packet loss were exploited to provide estimators for the loss probability per link in multicast trees from end-to-end measurements~\cite{presti2002multicast, mincloss}.

\subsection{Contributions}
The goal of this work is to bridge quantum tomography with classical network tomography by introducing the problem of Quantum Network Tomography. Our contributions are as follows.
\begin{itemize}
    \item We formally define the problem of Quantum Network Tomography for generic quantum networks as the estimation of probabilities determining operator-sum representations for all quantum channels in a network. The problem captures the end-to-end estimation of parameters by considering estimators based on measurements performed exclusively in the end-nodes.
    \item We define a generic multi-party state distribution process for the network tomography of arbitrary trees.
    \item We solve quantum network tomography problems in star networks using the multi-party state distribution process defined.
    In particular, we focus on the case where channels are characterized by a single Pauli operator, such as a bit-flip channel, and provide estimators for network parameters with polynomial sample complexity. Our estimators show that, even for single Pauli channels, pre-shared entanglement may provide advantages for estimation in terms of parameter identifiability.
    \item Finally, we show that our estimators attain the Quantum Cramèr-Rao bound and compare their performance numerically for a particular star network.
\end{itemize}

The remainder of this article is structured as follows. In Section~\ref{sec:background}, we provide the necessary mathematical background to discuss our results. In Section~\ref{sec:tomography}, we define the problem of quantum network tomography and describe the multi-party state distribution process. We report our results for the tomography problem in the case of star networks with channels described by a single Pauli operator in Section~\ref{sec:stars}. Finally, we discuss our results and present concluding remarks in Section~\ref{sec:conclusion}.

\section{Background}\label{sec:background}

In this work, we consider quantum networks to be systems formed by the interconnection of quantum processing devices with quantum channels that allow for communication of quantum information. We do not assume any particular physical implementation of the underlying quantum channels, nor any choice of qubit platform. In addition, we use the term quantum processing device, or quantum processor, to abstract quantum computers, routers, switches and repeaters.
We represent a quantum network as a graph $G = (V, E)$, where the nodes in $V$ denote the quantum processors and the edges in $E$ represent the quantum channels interconnecting the nodes. In addition, the node set is 
partitioned as $V = V_E \cup V_I$, where $V_E$ and $V_I$ determine the sets of end- and intermediate nodes in the network.

Quantum processors are assumed capable of performing generic, unitary quantum operations on its quantum registers and perfect quantum measurements. This assumption implies that the only source of error in the network comes from the transmission of quantum information through the channels. Despite being unrealistic, assuming perfect quantum operations in the nodes is the initial step for the investigation of the estimation problems within the scope of this article.

A star quantum network is a system consisting of one intermediate node interconnecting $n$ end-nodes. In this particular case, we have $|V| = n + 1$, $|E| = n$ and each end-node identifies one edge of the star. For simplicity, we label the intermediate node as node $n$ and edge $(v, n)$ as the $v$-th edge, for $v \in \{0, 1, \ldots, n - 1\}$, \textit{i.e} $V_E = \{0, 1, \ldots, n - 1\}$  and $V_I = \{n\}$.

In what concerns notation, we use $\hilb^{K}$ to represent the Hilbert space formed by $K$ qubits, \textit{i.e} the Hilbert space of dimension $2^{K}$, and $\hilbdm^{K}$ to represent the space of density matrices of $K$-qubit systems.
We represent $n$-qubit Greenberger–Horne–Zeilinger (GHZ) states as
\begin{align}
    & \ket{\Phi_{s}^{b}} = \frac{\ket{0s} + (-1)^{b} \ket{1\overline{s}}}{\sqrt{2}}, \label{eq:ghz}
\end{align}
where $s \in \{0, 1\}^{n - 1}$  and $b \in \{0, 1\}$. GHZ-basis projectors are defined as $\Phi_{s}^{b} = \dyad*{\Phi_{s}^{b}}$. We use the standard notation of $[.,.]: [A,B] = AB - BA$ and $\{.,.\}: \{A, B\} = AB + BA$ to respectively denote the commutator and anti-commutator operators.

\subsection{Quantum channels in the network}

The edges in $|E|$ represent Completely Positive Trace-Preserving (CPTP) single-qubit maps. We assume that, for every edge $e = (u, v) \in E$,
the quantum channel $\mathcal{E}_v: \hilbdm^{2} \to \hilbdm^{2}$ that interconnects nodes $u$ and $v$ corresponds to the mapping
\begin{align}
    & \mathcal{E}_{e}(\rho) = \sum_{k = 0}^{d_e - 1} \theta_{ek} U_{ek} \rho U_{ek}^{\dagger} \label{eq:gench},
\end{align}
where $\rho$ is a one-qubit density matrix, $\{U_{ek}\}$ is a set of $d_e \in \mathbb{N}$ unitary operators and
$\theta_{e} \in \mathbb{R}^{d_e}$ is a probability vector. We assume the channel models are known  in the sense that the set of operators $\{U_{ek}\}$ is known for all channels. In precise terms, we assume 
prior knowledge on a set of unitaries for which there exists a physically valid operator-sum representation of
channel $e$ in the form of \eqref{eq:gench}, for all $e \in E$.
In the case of the generic depolarizing channel, \eqref{eq:gench} reduces to
\begin{align}
    & \ch_e(\rho) = \sum_{k = 0}^{3} \theta_{ek} \sigma_{k} \rho \sigma_{k}, \label{eq:gendepch}
\end{align}
where $\theta_e \in \mathbb{R}^{4}$ and $\{\sigma_k\}$ denote the set of Pauli operators, with $\sigma_0 = I$.
When channels are described by a single Pauli operator,
\eqref{eq:gench} further reduces to 
\begin{align}
    & \mathcal{E}_e(\rho) = \theta_{e} \rho + (1 - \theta_{e}) \sigma \rho \sigma \label{eq:single_pauli}
\end{align}
where $\sigma \in \{X, Y, Z\}$ is one of the Pauli operators and the channel is described by a single parameter $\theta_e \in \mathbb{R}$. The description in \eqref{eq:gench} is general as it captures any physical process where a qubit state is indeed transmitted from a node to another.


\subsection{Quantum estimation}

Quantum estimation theory addresses how measurements can be used to obtain information
from quantum systems. There are three fundamental aspects of any quantum estimation problem. Assume that it is of interest to estimate a parameter vector $\theta$. First, it is necessary to obtain a quantum state ensemble $\rho$ that is parameterized by $\theta$, what is normally referred to as parameterization. Second, it is necessary to measure $\rho$ to obtain measurement statistics. Finally, we need to design estimators for $\theta$ based on measurement outcomes.

Suppose that we are given a parameterized state ensemble of $N$ qubits. In particular, let $\rho: \mathbb{R}^{M} \to \hilbdm^{N}$ denote the density matrix of the system that depends on a parameter vector $\theta \in \mathbb{R}^{M}$ as
\begin{align}
    & \rho(\theta) = \sum_{k = 0}^{R - 1}\lambda_k(\theta) \Lambda_k(\theta), \label{eq:dm}
\end{align}
where $R$ is the rank of $\rho(\theta)$, $\{\Lambda_k\}$ is the set of projectors for the eigenspace of $\rho(\theta)$  and  $\lambda_k: \mathbb{R}^{M} \to \mathbb{R}$ are $\theta$-dependent probability values for which $\sum_{k}\lambda_k(\theta) = 1$. We are interested in the following quantum estimation problem.

\begin{problem}[Quantum parameter estimation]\label{prob:quantumest}
Find an estimator $\thetahat$ for $\theta$ from measurements performed in an ensemble $\rho(\theta)$. 
\end{problem}

The measurement statistics used to describe $\thetahat$ will depend on the measurement performed. For any set of \textit{Positive Operator-Valued Measure} (POVM) elements $\{\Pi_l\}$, the probability of having outcome $l$ as a measurement result is
\begin{align}
    & p_{\theta}(\Pi_l) = \sum_{k = 0}^{R - 1}\lambda_k(\theta) \Tr[\Lambda_k(\theta) \Pi_l]. \label{eq:mes_prob}
\end{align}
Thus, if the form of $p_{\theta}$ is known, measuring an ensemble of states described by $\rho(\theta)$ with $\{\Pi_l\}$ yields estimators $\hat{p}_{\theta}(\Pi_l)$, and $\thetahat$ can be obtained by solving the inverse problem
\begin{equation}
    \begin{cases}
        & p_{\hat{\theta}}(\Pi_l) = \hat{p}_{\theta}(\Pi_l) \text{, for all } l, \\
        & \norm*{\hat{\theta}}_1 = 1.
    \end{cases} \label{eq:inv_est_prob}
\end{equation}

There are two aspects of such estimation problem. First, the number of equations obtained is the number of POVM elements. For projective measurements, the number of equations grows as $\mathcal{O}(2^N)$ because of the completeness relation. Second, the dependence of $p_{\hat{\theta}}(\Pi_l)$ on $\theta$ does not guarantee that the inverse problem in \eqref{eq:inv_est_prob} has a unique solution.

\subsection{Estimation efficiency}

In general, different parametrization processes provide mixed states described by different density matrices, which leads to distinct estimators $\thetahat$ for $\theta$. In addition, different POVMs yield distinct estimators for the same density matrix $\rho(\theta)$. Thus, we analyze the efficiency of different estimators for  Problem \ref{prob:quantumest} with two metrics of interest: the \textit{identifiability} of $\hat{\theta}$ and the \textit{Quantum Fisher Information Matrix} (QFIM) for $\theta$~\cite{liu2019quantum}. We say that an estimator $\hat{\theta}$ identifies $\theta$ if it determines a unique value for $\theta$ from a sequence of observations. In the case of \eqref{eq:inv_est_prob}, $\hat{\theta}$ identifies the parameters if it is the unique solution to the inverse problem. The QFIM $\mathcal{F}$ of a density matrix $\rho(\theta)$~\cite{helstrom1969quantum} is a semi-positive definite real matrix with entries
\begin{align}
    & \mathcal{F}_{jk} = \frac{1}{2}\Tr(\rho\{L_j, L_k\}) \label{eq:qfim},
\end{align}
where $L_k$ is the \textit{Symmetric Logarithmic Derivative} (SLD) of $\rho(\theta)$ with respect to $\theta_j$ given by the differential matrix equation
\begin{align}
    & \frac{\partial \rho}{\partial \theta_j}  = \frac{1}{2} (L_j\rho + \rho L_j). \label{eq:sld}
\end{align}
The diagonal entry $\mathcal{F}_{jj}$ contains the \textit{Quantum Fisher Information} (QFI) for parameter $\theta_{j}$. The QFIM yields the \textit{Quantum Cramèr-Rao bound} (QCRB) for multi-parameter estimation, which is a lower bound for the covariance matrix of estimators based on statistics generated by any POVM~\cite{liu2019quantum}. An estimator that reaches the QCRB for the joint estimation of $\theta$ is attainable if, and only if $L_j$ and $L_k$ commute for all pairs $\theta_j, \theta_k$~\cite{liu2019quantum}. Furthermore, consider the following simple theorem.

\begin{theorem}\label{th:slddiag}
    If $\rho(\theta)$ is diagonalized by a set of $\theta$-independent projectors $\{\Lambda_k\}$, then $\{\Lambda_k\}$ diagonalizes $L_j$, for all $j$. 
\end{theorem}
\begin{proof}
    Under the assumption that $\Lambda_k(\theta) = \Lambda_k$, $\rho(\theta) = \sum_{k} \lambda_k(\theta) \Lambda_k$ and
    $\frac{\partial \rho(\theta)}{\partial \theta_j} = \sum_{k} \frac{\partial \lambda(\theta)}{\partial \theta_j} \Lambda_k$ for all $j$. Thus, by taking the ansatz $L_j = \sum_{k} l_{jk} \Lambda_k$, one verifies that \eqref{eq:sld} is solved with $l_{jk} = \frac{\partial \lambda(\theta)}{\partial \theta_j}\frac{1}{\lambda(\theta)}$, for all $j$.
\end{proof}

Furthermore, when Theorem \ref{th:slddiag} is valid, $L_i$ and $L_j$ commute for all $i, j$ and an estimator $\thetahat$ from projective measurements on the $\{\Lambda_k\}$ basis is asymptotically optimal given $\rho(\theta)$. Plugging the description for $L_j$ from Theorem \ref{th:slddiag} in \eqref{eq:qfim} yields
\begin{align}
    & \mathcal{F}_{ab} = \sum_{k} \frac{1}{\lambda_k}\frac{\partial \lambda_k}{\partial \theta_a}\frac{\partial \lambda_k}{\partial \theta_b}. \label{eq:QFIM}
\end{align}

\section{Quantum Network Tomography}\label{sec:tomography}

\begin{figure*}
\begin{centering}
\begin{subfigure}[b]{0.33\textwidth}
    \centering
    \includegraphics[scale=0.16]{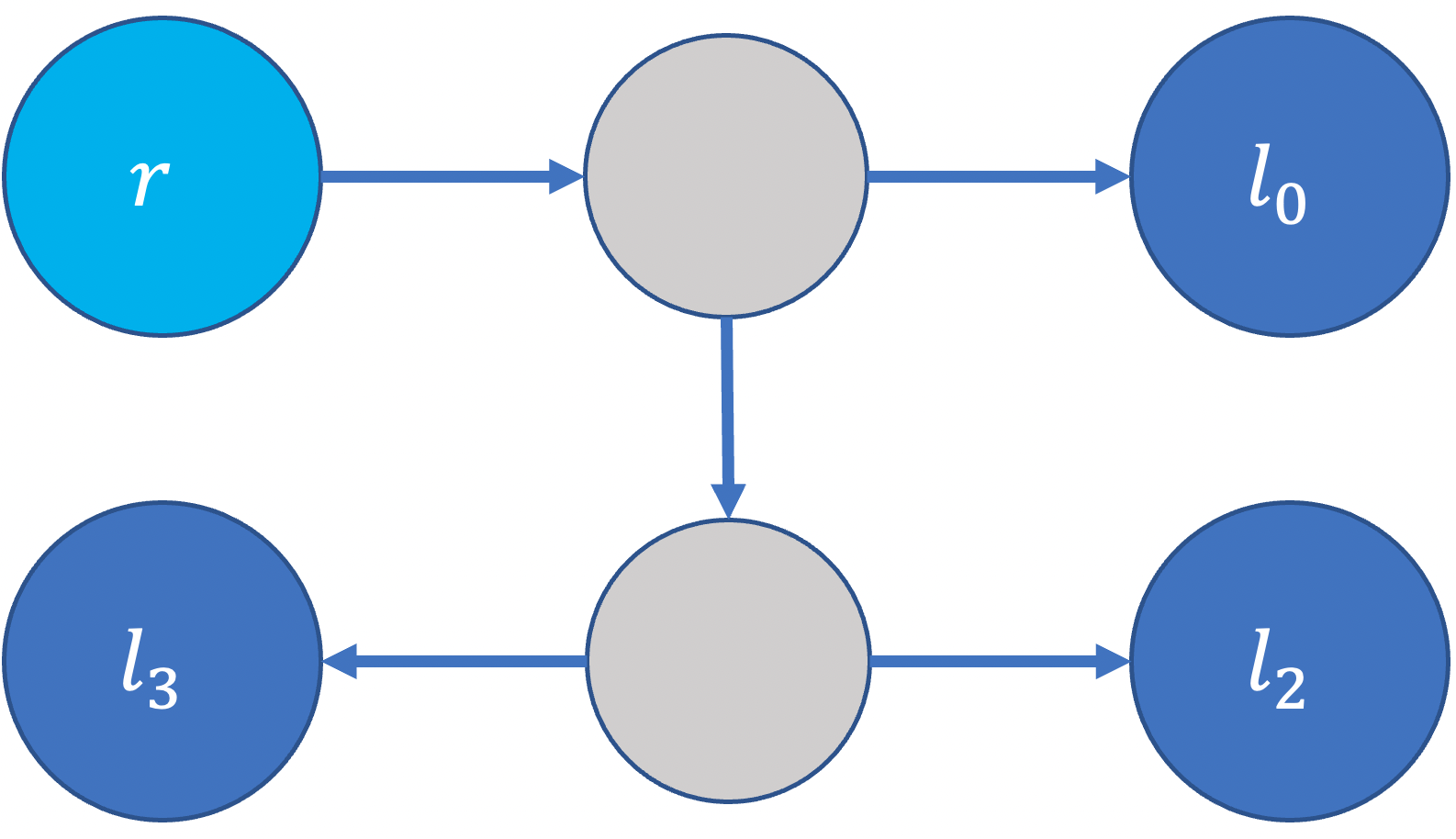}
    \caption{State distribution.}
    \label{subfig:dist}
\end{subfigure} \hfill
\begin{subfigure}[b]{0.33\textwidth}
    \centering
    \includegraphics[scale=0.16]{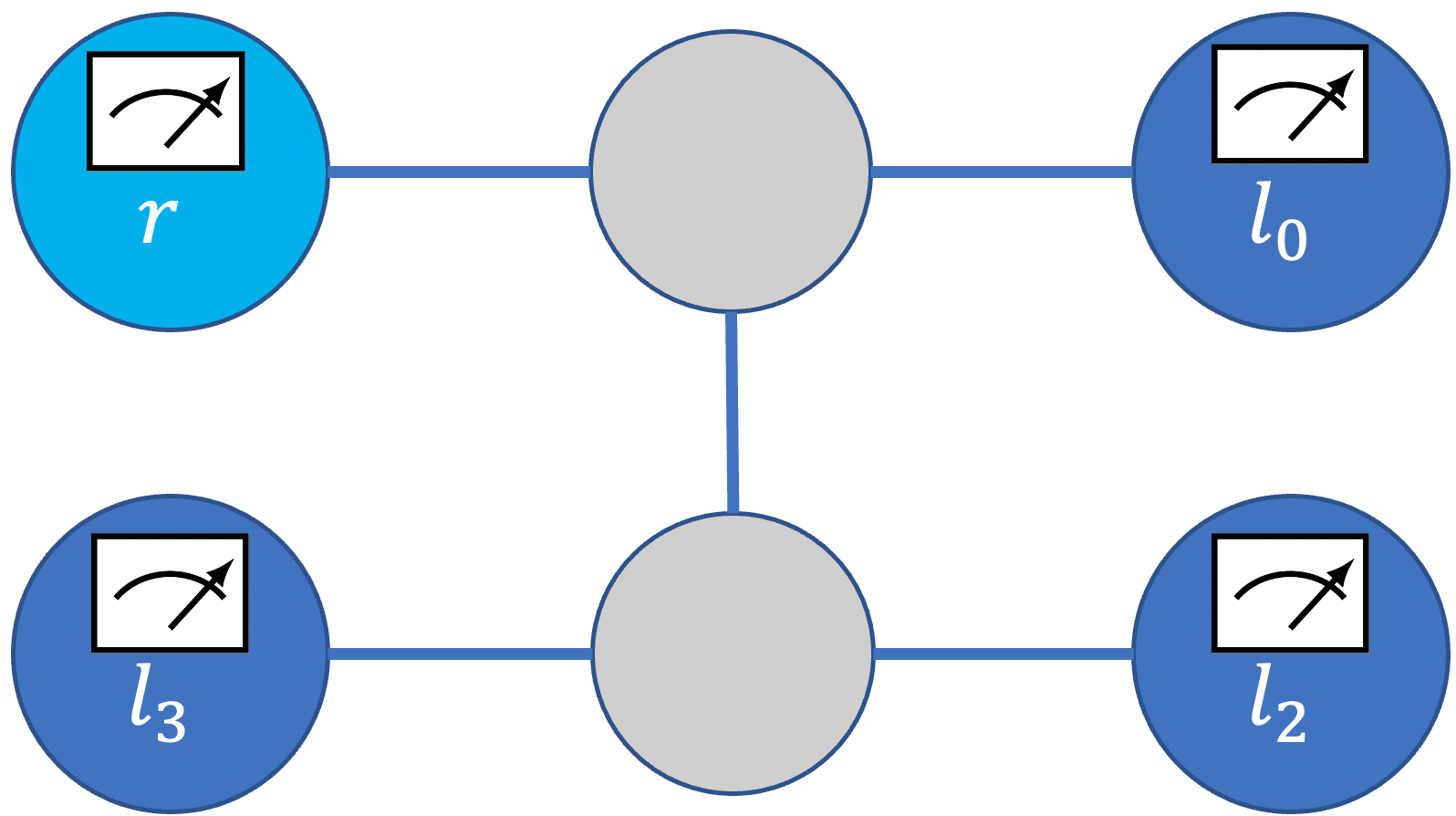}
    \caption{End-to-end measurements.}
    \label{subfig:measurements}
\end{subfigure} \hfill
\begin{subfigure}[b]{0.27\textwidth}
    \centering
    \includegraphics[scale=0.16]{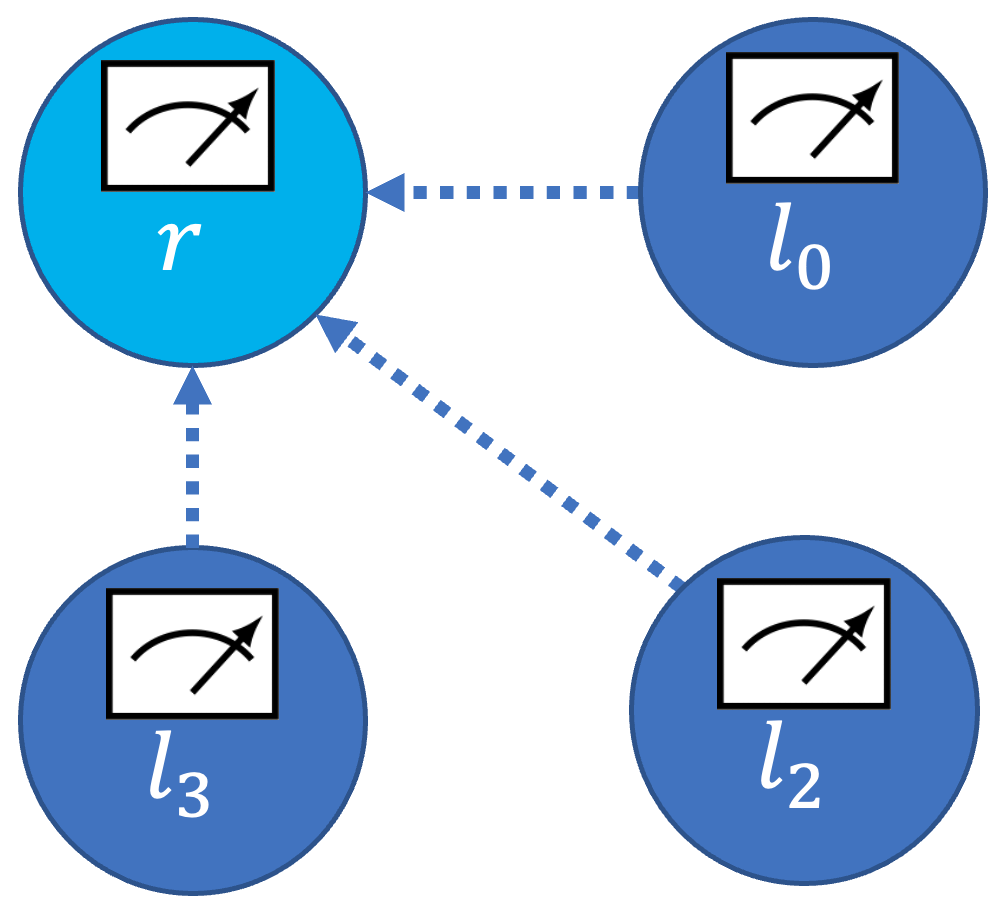}
    \caption{Collection and estimation.}
    \label{subfig:collection}
\end{subfigure} \hfill
\end{centering}
\caption{Quantum network tomography on trees. The tree shown is formed by two intermediate nodes and each channel is described by a Kraus decomposition of known form (Equation \ref{eq:gench}). The goal is to estimate the probabilities describing each channel in the tree. (\subref{subfig:dist}) A mixed state $\rho(\theta)$ is distributed from the root (node $r$) towards the leaves $l_j$. To perform distribution, the root and the intermediate nodes (small gray circles) transmit a qubit state to its neighbors following arrow direction. (\subref{subfig:measurements}) When all leaves receive qubits, measurements are performed to obtain classical data for estimation. (\subref{subfig:collection}) Finally, end-nodes transmit the classical data from measurements to a fusion center in the network, represented here by the root node, for estimation to be performed.}
\label{fig:tomography}
\end{figure*}

In this section we rephrase Problem \ref{prob:quantumest} for a networking setting.
We begin by describing the problem of channel tomography, which naturally leads to the definition of the network version. Consider the simplest non-trivial network system formed by two end-nodes $u$ and $v$ connected by an edge $e = (u, v)$ representing the quantum channel $\mathcal{E}_e$ following \eqref{eq:gench}. The channel tomography problem refers to quantum process tomography of $\ch_e$. In it's most general version, one does not assume knowledge of the form of the channel in terms of Kraus operators that yield a physically valid operator-sum representation of $\ch_e$. Then, the solution for this general problem is a set of unitaries $\{U_{ek}\}$ and an estimator $\thetahat_e$ that characterize the operator-sum representation.

This general problem is difficult, since the unitaries must be characterized from measurements, treating $\ch_e$ as a black-box. We simplify this problem by assuming that we know the set of unitary operators for which a valid operator-sum representation exists and focus on the estimation of $\theta$. Thus, the formal definition of the channel tomography problem of interest is as follows.

\begin{problem}[Quantum Channel Tomography]\label{prob:chtomo}
    Given a set of $M$ unitary operators $\{U_{k}\}_e$ for which the quantum channel $\ch_e$ has the form given in \eqref{eq:gench}, estimate a probability vector $\thetahat_e \in \mathbb{R}^{M}$ that characterizes $\ch_e$.
\end{problem}

The channel tomography problem is an instance of Problem \ref{prob:quantumest} since, in order to estimate the probability vector, it is necessary to use the channel to prepare and measure an ensemble of mixed states that depends on $\theta$ to obtain statistics for estimation. In this case, $u$ and $v$ are end-nodes and can perform arbitrary measurements on the ensemble, share the classical results of the measurements and compute $\thetahat_e$.

The quantum network tomography problem, which is depicted in Figure \ref{fig:tomography}, extends the channel problem to networks with the caveat that only end-nodes can contribute information for estimation. We now present the formal definition of the quantum network tomography problem for trees, which is one of the main contributions of this article. A tree $T = (V, E)$ is a connected graph with no cycles.

\begin{problem}[Quantum Network Tomography]\label{prob:gentomograpy}
Given a quantum network with topology given by a tree $T = (V, E)$, with nodes partitioned into disjoint sets $V_E$ and $V_I$ and a set of unitary operators $\{U_k\}_e$ for each $e \in E$, estimate a probability vector $\thetahat_e$ characterizing $\ch_e$ for all $e \in E$ using measurement statistics from nodes in $V_e$ exclusively.
\end{problem}

The main difference between Problems \ref{prob:chtomo} and \ref{prob:gentomograpy} is the restriction that measurement observations used for estimation must come from end-nodes in the network case. If such a restriction is dropped, the network problem reduces to the channel problem and the best way to obtain estimators is to treat each channel independently and solve Problem \ref{prob:chtomo} for all channels in the network. When the restriction is considered, joint estimation of the channel parameters must be carried out in the general case. The joint estimation reflects the fact that channels must be jointly used to prepare an ensemble of states described by $\rho(\theta)$ that can be measured by the end-nodes to provide statistics for estimation.

\subsection{Parameterization and State Distribution}

With the formal definition of the tomography problem, we propose a family of parameterization processes that cast Problem \ref{prob:gentomograpy} into an instance of Problem \ref{prob:quantumest}.
In principle, the preparation of $\rho(\theta)$ must target qubits in the end-nodes for measurements to be performed. It is straightforward that such preparation must use the channels in the network in order to make the ensemble dependent on $\theta$. Regardless of how channels are used to prepare the ensemble, the parameterization process can be expressed as an abstract, multi-qubit quantum channel $\mathcal{N}_{\theta}$ acting on a locally-prepared ensemble $\rho_0$ as
\begin{align}
    & \rho(\theta) = \mathcal{N}_{\theta}(\rho_0) \label{eq:abstract_tomo}.
\end{align}
In our terminology, a locally-prepared state is any state that is separable with respect to the nodes of the network such that $\rho_0$ follows
\begin{align}
    \rho_0 = \bigotimes_{v \in V}\rho_{0v} \label{eq:localstate},
\end{align}
where $\rho_{0v}$ is a multi-qubit state located in a quantum register in node $v$. Following the quantum process tomography nomenclature, we refer to $\rho_0$ from now on as the probe state. The key aspect of quantum network tomography is that both $\rho_0$ and $\mathcal{N}_{\theta}$ are allowed to be chosen as part of the problem's solution.

We describe the abstract quantum channel $\mathcal{N}_{\theta}$, and hence the parameterization process, with a process for network state distribution. The state distribution problem refers to the preparation of generic quantum states across target end-nodes, using channels and intermediate nodes to propagate entanglement. In this setting, nodes start with a state that follows \eqref{eq:localstate}, which is progressively transformed through local operations and quantum state transmissions across channels. State distribution is a natural approach to define $\mathcal{N}_{\theta}$ because qubits transmitted across a link $e$ evolve according to $\mathcal{E}_e$ and, thus, can be used to incorporate $\{\theta_{ek}\}$ into the density matrix describing the distributed state. This process yields a quantum channel $\mathcal{N}_{\theta}$ that is a composition of $\{\mathcal{E}_e\}$ for all $e \in E$ used for distribution of $\rho_0$ to the end-nodes.


We propose a distribution process for trees (Algorithm \ref{alg:statedist}) to distribute quantum states across the network with properties tailored for tomography. Since trees have no cycles, there exists exactly one path interconnecting any two nodes in the network. Furthermore, the process uses each link in the network to transmit a qubit among neighboring nodes exactly once. Such process is general in the sense that it captures any quantum state distribution operation across trees under the restriction that a single qubit is transmitted between the nodes for distribution.

\begin{algorithm*}[ht]
    \SetAlgoLined
    \SetStartEndCondition{ }{}{}
    \SetKwProg{Fn}{def}{\string:}{}\SetKwFunction{Range}{range}
    \SetKw{KwTo}{in}\SetKwFor{For}{for}{\string:}{end}
    \SetKwIF{If}{ElseIf}{Else}{if}{:}{elif}{else:}{}
    \SetKwFor{While}{while}{:}{fintq}

    \SetKwInOut{Input}{input}\SetKwInOut{Output}{output}

    \Input{tree $T$; circuit $\mathcal{C}$; function $\eta$}
    \Output{distributed state $\rho(\theta)$ across ${r} \cup L_{T}$}

    \BlankLine
    \For{$k \in \{0, 1, \ldots, h_{\max}\}$}{
       \For{$v \in H(k)$}{
            $v$ receives qubit in state $\mathcal{E}_{P_v}(\rho_v)$ from $P_v$\;
            $v$ performs circuit $\rho_{v'} := \mathcal{C}_v(\dyad{0_{n_{v}}} \otimes \mathcal{E}_{P_v}(\rho_v))$\;
            \For{$u \in S_{v}$}{
            $v$ transfers qubit state indexed by $\eta(v, u)$ in state $\rho_{u}$ to $u$\;
            }
        }      
    }
    \caption{Tree state distribution}\label{alg:statedist}
\end{algorithm*}

Thus, consider the following definitions. Let $T = (V, E)$ be a tree rooted in node $r$. The height function $h: V \to \mathbb{Z}^{+}$ is defined such that $h(v)$ is the hop-distance of the path connecting $v$ and $r$ in $T$. Let $h_{\max}= \max_v h(v)$. A leaf $v$ of $T$ is a node with height $h_{\max}$. Let the set $L_T = \{v: h(v) = h_{\max}\}$ denote the set of leaves of $T$ and $H(k) = \{v: h(v) = k\}$ denote the set of all nodes of $T$ with height $h$. Let the predecessor of 
$v$ be the neighbor $P_v$ of $v$ with $h(P_v) = h(v) - 1$, which is not defined for $r$. Let the successor set of $v$ be $S_v = \{u : (v, u) \in E \text{ and } h(u) = h(v) + 1\}$ for all $v$ in $T \setminus L_T$, and $S_v = \emptyset$ for all $v \in L_T$. Furthermore, let $\mathcal{C}$ denote the description of a quantum circuit applied in the nodes of the network such that $\mathcal{C}_v$ is a generic, multi-qubit circuit applied in node $v$. Given the circuit $\mathcal{C}_{v}$, let $n_v$ denote the number of qubits in $\mathcal{C}_{v}$. Let $\eta: V \times V \to \mathbb{Z}^{+}$ be a function determining the index of qubits to be transmitted between neighbors. $\eta(u, v)$ determines the index of the output qubit from $\mathcal{C}_u$ to be transmitted from $u$ to $v$ after $\mathcal{C}_u$ is performed. Finally, let $\dyad{0_n}$ denote the pure state $\dyad{0}^{\otimes n}$.

The process in Algorithm \ref{alg:statedist} is understood as follows. The inputs are the rooted tree $T$, a quantum circuit description $\mathcal{C}$ and a qubit index function $\eta$. In order to simplify the process description, we assume that the initial local state $\rho_0$ is encoded in the circuit description $\mathcal{C}$, such that all nodes start with registers prepared in a pure state of form $\dyad{0_n}$. The process begins with the root note executing $\mathcal{C}_{r}$. Note that, since the root has no predecessor, line 3 has no effect for $v = r$. Then, the root sends one qubit from the output of $\mathcal{C}_r(\dyad{0_{n_r}})$ to each one of its successors, and keeps $n_{r} - |S_r|$ qubits in memory. The function $\eta(r, v)$ specifies which of the $n_{r}$ qubit states is transmitted from $r$ to $v$, for all neighbors $v$ of $r$. In the process description, $\rho_v$ denotes the state prepared in $S_v$ and sent to $v$.  Whenever node $v$ receives a qubit, the channel interconnecting $P_v$ and $v$ transforms the initially transmitted state $\rho_v$ into $\mathcal{E}_{P_v}(\rho_v)$. Thus, node $v$ applies the quantum circuit $\mathcal{C}_v$ to state $\dyad{0_{n_{v}}} \otimes \mathcal{E}_{P_v}(\rho_v)$ as input. The node proceeds by transmitting one qubit to each one of its successors following $\eta$, again keeping $n_{v} - |S_v|$ qubits in memory. The process terminates when all leaves receive a qubit state from its predecessors. It is clear from line 6 of Algorithm \ref{alg:statedist} that each link in $T$ is used exactly once for distribution.

The generality of the process described in Algorithm \ref{alg:statedist} stems from the freedom of defining both $\mathcal{C}$ and $\eta$, although it should be clear from its description that the distributed mixed state depends on $\theta_{ek}$, for all $e \in E$, as long as $\mathcal{C}_{v}$ acts on $\dyad{0_v} \otimes 
\mathcal{E}(\rho_v)$ with non-separable quantum operations.

There exists a direct mapping determining the form of $\mathcal{N}_{\theta}$ from $T$, $\mathcal{C}$ and $\eta$, although writing the mathematical form for the general case is lengthy.
In the remainder of this article, we consider distribution circuits with $n_v = |S_v| + \delta_{vr}$ for all nodes $v \in V$, where $\delta_{vr}$ is the discrete delta function, what implies that no qubits remain in the intermediate nodes after distribution. Furthermore, we focus on the description of $\rho(\theta)$ directly rather then explicitly writing $\mathcal{N}_\theta$.
\section{Tomography in star networks}\label{sec:stars}

In this section we specialize our approach to star quantum networks. The star graph is a simple type of tree for which we can
observe the workings of the distribution process and study tomography in detail. We consider a particular
bi-partition of the nodes, where all leaves of the star are end-nodes and the intermediate node is the
center node, as depicted in Figure \ref{fig:star}. Following our definitions from Section \ref{sec:background}, this implies that $V_E = \{0, \ldots, n - 1\}$
and $V_I = \{n\}$ for a $(n + 1)$-node star graph.

Star graphs are simple in the context of Problem \ref{prob:gentomograpy} because the distance between end-nodes
is always two. In addition, we simplify the problem further by considering a scenario where all the channels
are described by a single Pauli operator following \eqref{eq:single_pauli}. This simplification helps both
the description and evaluation of estimators without rendering the problem trivial. Such class of channels
suffices to demonstrate the difficulty of performing tomography in a network and builds intuition on how
to approach the problem for more complex channels like the generic depolarizing channel, which is described by all
Pauli operators. We describe our methods for
pure bit-flip channels and discuss how they generalize to channels described by the other two Pauli
operators. Note that we consider the same Pauli operator to describe all channels in the star.

\begin{figure}
    \centering
    \includegraphics[scale=0.25]{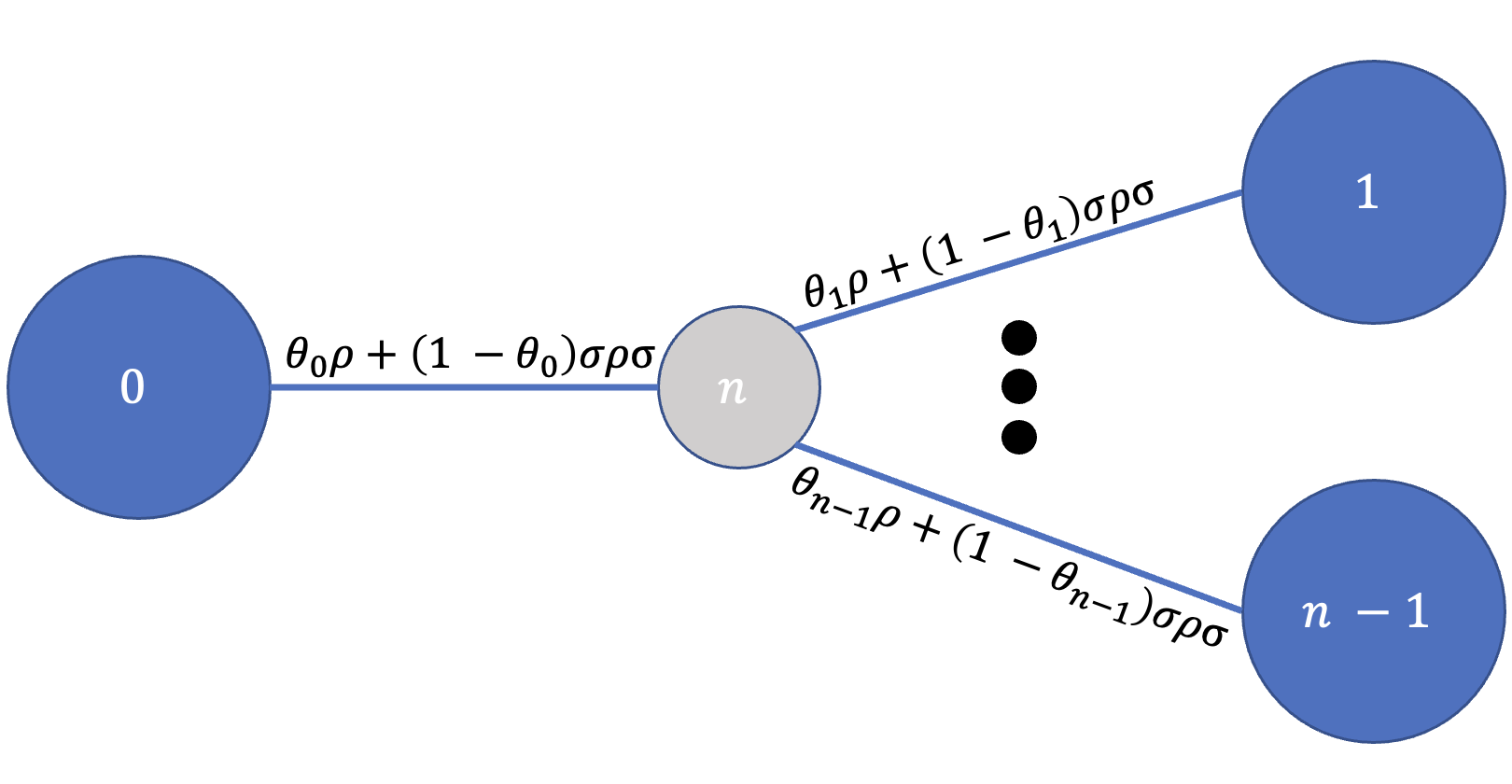}
    \caption{Star quantum networks with $n$ end-nodes depicted as large blue circles. Each channel $j$ in the network is described by a Pauli operator $\sigma$ and a probability $\theta_j$. The operator-sum representation for the channels is written along the edges of the network.}
    \label{fig:star}
\end{figure}

\subsection{Tomography in the basis of Pauli operators}

We start by describing a solution for the tomography in the star that uses states in the Pauli basis. Under the
assumption of bit-flip channels, we target states in the $Z$ basis. The same analysis follows for the other
single Pauli channels by selecting a different basis according to the Pauli operator under consideration.
In the case of $Z$ we use the $X$ basis and in the case of $Y$ we can use either the $Z$ or the $X$ basis.

We are interested in the following distribution process based on Algorithm \ref{alg:statedist}. The root prepares
a single qubit in state $\ket{0}$. Since the distribution process assumes that qubits are initialized in state $\ket{0}$,
the root applies the identity operator as a circuit, \textit{i.e} $\mathcal{C}_0 = I$. The root transmits the state
to the intermediate node, which receives the mixed state
\begin{align}
    & \ch_0(\rho_n) = \theta_0 \dyad{0} + (1 - \theta_0) \dyad{1}.
\end{align}
The intermediate node applies the generalized Toffoli gate
\begin{align}
    & T_n = \dyad{0} \otimes I^{\otimes n - 1} + \dyad{1} \otimes X^{\otimes n - 1}
\end{align}
controlled by the qubit it received on $n - 1$ qubits in its quantum register such that
$\mathcal{C}(n) = T_{n - 1}$, what yields the mixed state
\begin{align}
    &  \theta_0 \dyad{0}^{\otimes n - 1} + (1 - \theta_0) \dyad{1}^{\otimes n - 1},
\end{align}
The intermediate node assigns its qubits to end-nodes following the order of node labels,
such that the qubit indexed by $j$ is sent to node $j + 1$, with $0 \leq j \leq (n - 2)$.
In terms of the inputs of the distribution process, the mapping function selected in the
intermediate node is $\eta(n, j + 1) = j$ for $j \in \{0, 1, \ldots, n - 2\}$. The final
mixed state received by the end-nodes is the $(n - 1)$-qubit state
\begin{align}
    & \rho(\theta) = \sum_{s \in \mathbb{B}^{n - 1}} \alpha(s) \dyad{s},\label{eq:dmstrings}
\end{align}
where
\begin{multline}
    \alpha(s) = \theta_0 [\prod_{j = 1}^{n - 1} \delta_{\overline{s}_j} \theta_j + \delta_{s_j} (1 - \theta_j)] \\
    + (1 - \theta_0) [\prod_{j = 1}^{n - 1} \delta_{s_j} \theta_j + \delta_{\overline{s}_j} (1 - \theta_j)] \label{eq:stringprob}
\end{multline}
where $\delta_{s_j}$ is the discrete pulse function equal to $1$ if the bit $s_j = 1$ and $\overline{s}_j= 1 - s_j$. In this case, the final density matrix spams all of the binary strings with $n - 1$ bits and is
diagonal on the $Z$ basis.

Since the density matrix in \eqref{eq:dmstrings} is diagonal on the $Z$ basis and only
the eigenvalues depend on the vector $\theta$, it is easy to verify that the SLD $L_j$ for
each parameter $\theta_j$ is also diagonal on the Z basis. Thus, in the context of the
QCRB, the POVM of choice to measure $\rho(\theta)$ is the set of projective measurements
on the $Z$ basis for the Hilbert space of  $n - 1$ qubits, which is attained by
local measurements of each qubit in the $Z$ basis.

Given that local measurements are performed by the end-node,
we can write the statistics of flips in each particular bit as follows. Let
$S_j\in \{0,1\}$ denote the measurement outcome
of the qubit in node $j$. A bit-flip is measured in node $j$ if a flip occurs exclusively on one of the  channels $0$ and $j$. Let $F_j\in \{0,1\}$ denote the absence or presence of a bit flip on channel $j$.  We have
\begin{align}
    & S_j = F_0 \oplus F_j 
\end{align}
where $\oplus$ denotes the XOR operation. Thus, the probability of measuring a bit-flip in qubit $j$ is
\begin{align}
    & Pr[S_j = 1] = \theta_0 (1 - \theta_j) + (1 - \theta_0) \theta_j \label{eq:bitflip},
\end{align}
for all $j \in \{1, \ldots n - 1\}$. Using \eqref{eq:bitflip} for all channels in the star
yields a system of $n - 1$ first-order, bi-variate polynomial equations over $n$ variables. However,
all of these equations depend on $\theta_0$ and, if $\theta_0$ can be computed, the
system reduces to a system of $n - 1$ independent linear equations.

The dependency between $S_j$ and $F_0$ introduces dependencies between all pairs of variables
$S_j, S_k$ for all cases where $\theta_0 \neq 0.5$. This dependency can be exploited to obtain
an equation for $\theta_0$ as follows. Let $S_{jk} = S_j S_k$ denote the joint random variable obtained 
by concatenating $S_j$ and $S_k$. Also, let $Pr[S_j = 1] = p_j$ and $Pr[S_{jk} = 11] = p_{jk}$.
The probability that any two flips are jointly measured by end-nodes $j$ and $k$ is
\begin{align}
    & p_{jk} = \theta_0 (1 - \theta_j) (1 - \theta_k) + (1 - \theta_0) \theta_j \theta_k. \label{eq:jointflip}
\end{align}
The probability in \eqref{eq:bitflip} can be re-arranged to obtain
\begin{align}
    & \theta_j = \frac{p_j - \theta_0}{1 - 2 \theta_0}, \label{eq:system}
\end{align}
which is valid for all $j \in \{1, 2, \ldots, n - 1\}$. Plugging back on \eqref{eq:jointflip} yields
the quadratic equation
\begin{align}
    &  a_{jk}(1 - \theta_0)\theta_0 + c_{jk} = 0 \label{eq:theta0}
\end{align}
for $\theta_0$, where
\begin{align}
    & a_{jk} = 1 + 4 p_{jk} - 2(p_j + p_k), \\
    & c_{jk} = p_jp_k - p_{jk}.
\end{align}

The form of \eqref{eq:theta0} is symmetric with respect to probabilities since
if $\theta^{*}$ solves the equation, $(1 - \theta^{*})$ also does. This inherent symmetry
implies that solving \eqref{eq:theta0} for a specific pair of end-nodes $(j, k)$ determines
two possible values for $\theta_0$ that are valid for the measurement results. More interestingly,
the symmetry cannot be broken even when considering \eqref{eq:theta0} for all pairs of end-nodes.

Finally, combining the system from \eqref{eq:bitflip} with \eqref{eq:theta0} gives two vectors
$\thetahat$ for the channel parameters that are compatible with the observations. Following our characterization,
the estimators obtained from this method do not identify the parameters completely, since we have two possible values of
$\thetahat$. Given the form of the solutions, identifiability can be obtained by assuming either low or high noise
regime for $\theta_0$. In this case, it suffices to select the solution of \eqref{eq:theta0} that is greater than $0.5$ in the low noise regime and the one the is smaller in the high one.


We simulate a four-node star with bit-flip channels characterized by $\theta = [0.8, 0.3, 0.4]$ to demonstrate the estimators using $Z$ basis measurements. We plot the estimated value for the three parameters with respect to the number of measurement outcomes used for estimation in Figure \ref{fig:z}. The symmetry in the estimators appears in the form of the two curves obtained for each parameter. The two values obtained for $\theta$ are $[0.8, 0.3, 0.4]$ and $[0.2, 0.7, 0.6]$, and identifiability can only be achieved by making an assumption on the error model.

\begin{figure}
    \centering
    \includegraphics[scale=0.46]{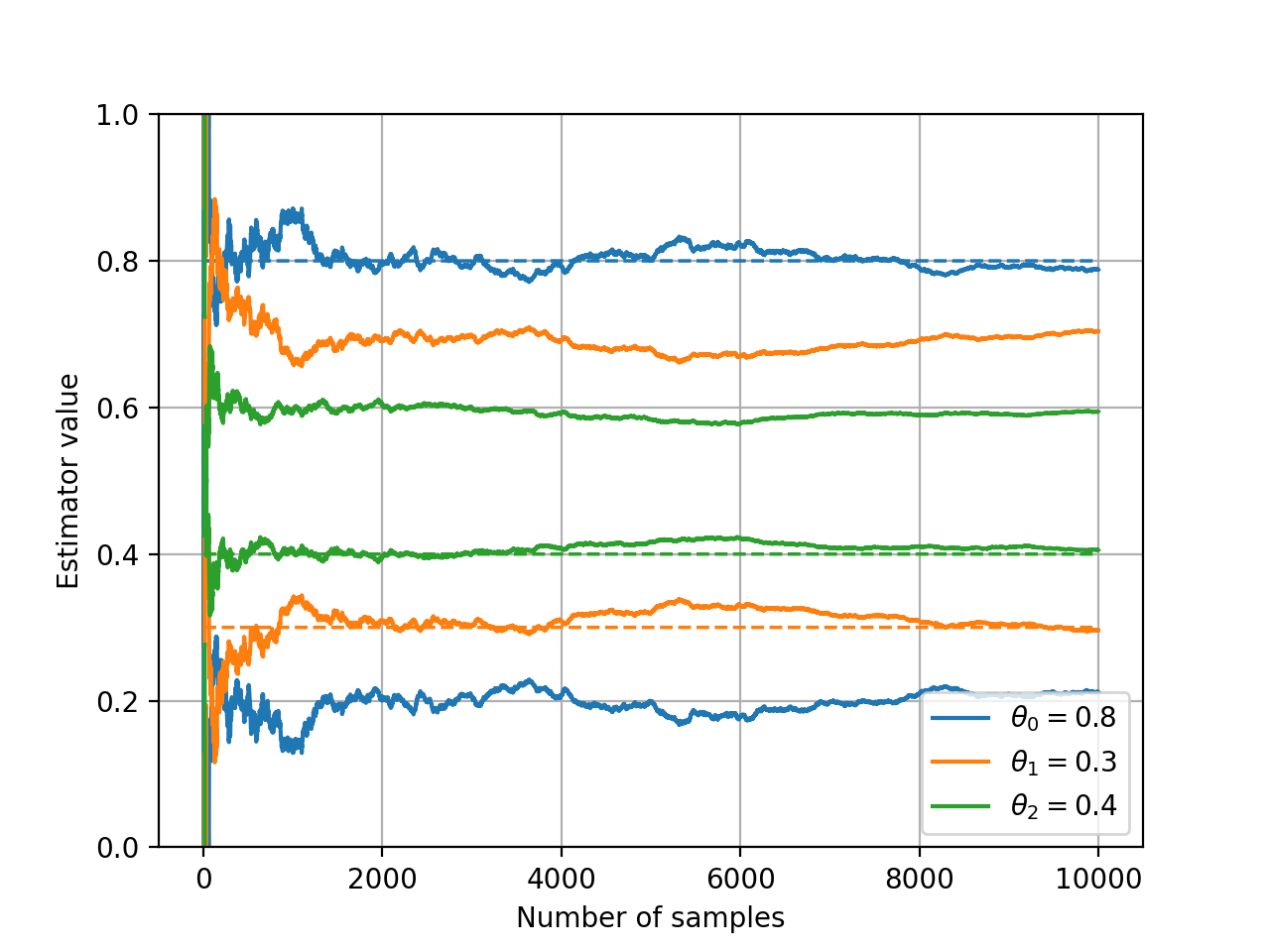}
    \caption{Estimators based on outcomes from $Z$ basis measurements for a four-node star graph with ground-truth parameter vector $\theta = [0.8, 0.3, 0.4]$. The dotted lines show ground-truth values for the parameters. The estimators cannot identify the parameters and there are two solutions compatible with observations.
    }
    \label{fig:z}
\end{figure}

\subsection{Tomography in the GHZ basis}

The need to introduce another assumption to obtain identifiability for the parameters motivates the search for other estimators.
We now proceed to describe how GHZ states can be used to address this issue. In particular, we define estimators that use global measurements in the end-nodes, which can be attained by pre-sharing entanglement among the end-nodes. One can argue that end-to-end entanglement is an important resource in a quantum network, and using global measurements introduces complexity in the implementation of our tomography process. Nonetheless, it is of interest to the scope of this work to analyze the benefits that entanglement may provide in the network tomography setting.

A GHZ state is a maximally entangled state that generalizes the Bell basis to more than two qubits. From \eqref{eq:ghz}, $n$ bits are necessary to describe an $n$-qubit GHZ state. GHZ states are interesting in this scenario because the state obtained after applying a Pauli operator on one qubit of a GHZ state is also a GHZ state. Formally, the state $\ket*{\Phi_{s}^{b}}$ evolves under the application of a Pauli operator $\sigma$ on its $j$-th qubit as
\begin{equation}
    \sigma \ket*{\Phi_{s}^{b}} =
    \begin{cases}
        \ket*{\Phi_{s \oplus s_j}^{b}}, \text{ if } \sigma = X, \\
        i\ket*{\Phi_{s \oplus s_j}^{b \oplus 1}}, \text{ if } \sigma = Y, \\
        \ket*{\Phi_{s}^{b \oplus 1}}, \text{ if } \sigma = Z,
    \end{cases}
\end{equation}
where $s_j$ is a binary string with $1$ in position $j$ if $j > 0$ and string $11\ldots1$ if $j$ = 0.

The instance of the distribution process used previously to distribute a mixed state diagonal on the $Z$ basis can be modified to distribute a mixed state diagonal on the GHZ basis by simply changing the circuit applied by the root. Instead of sending state $\ket{0}$ to the intermediate node, the root prepares the Bell state $\ket*{\Phi_{0}^{0}}$ and sends the second qubit to the intermediate node. This is achieved by the circuit $\mathcal{C}_0 = [H \otimes I, \text{CNOT}]$, assuming that the CNOT gate is controlled by the first argument. When describing circuits with multiple gates, we use an ordered list notation $[, ]$ to indicate that gates are applied on the order they appear inside the square brackets. We select the second qubit to be transmitted just to simplify notation because \eqref{eq:ghz} uses the first qubit as the reference binary value in the GHZ state superposition. If this is indeed the only modification considered, the mixed state received by the intermediate node is
\begin{align}
    & I \otimes \ch_0(\Phi_{0}^{0}) = \theta_0 \Phi_{0}^{0} + (1 - \theta_0) \Phi_{1}^{0} \label{eq:interghz0}
\end{align}
and the final mixed state distributed is
\begin{align}
    & \rho(\theta) =  \sum_{s \in \mathbb{B}^{n - 1}} \alpha(s) \Phi_{s}^{0}, \label{eq:dmghz0}
\end{align}
with probabilities $\alpha(s)$ given by \eqref{eq:stringprob}.

By comparing \eqref{eq:dmstrings} and \eqref{eq:dmghz0}, there is no gain in using GHZ states. This is intuitively understood by considering the fact that only $n - 1$ bits of the GHZ state are used to parameterize the necessary information, which is the same amount of bits used in the initial case. Thus, the key to obtain parameter identifiability is to slightly modify the circuit applied by the intermediate node to transform equation \eqref{eq:interghz0} into the mixed state
\begin{align}
    &  \theta_0 \Phi_{0}^{0} + (1 - \theta_0) \Phi_{0}^{1} \label{eq:interghz}
\end{align}
before the intermediate node transmits. Departing from \eqref{eq:interghz}, the final mixed state distributed is described by the $n$-qubit density matrix
\begin{align}
    & \rho(\theta) = \sum_{s \in \mathbb{B}^{n - 1}} \theta_0 \beta_0(s) \Phi_{s}^{0} + (1 - \theta_0) \beta_1(s) \Phi_{s}^{1}, \label{eq:identifiable_dm}
\end{align}
where
\begin{align}
    & \beta_0 = \prod_{j = 1}^{n - 1} \delta_{\overline{s}_j} \theta_j + \delta_{s_j} (1 - \theta_j), \\
    & \beta_1 = \prod_{j = 1}^{n - 1} \delta_{s_j} \theta_j + \delta_{\overline{s}_j} (1 - \theta_j).
\end{align}

The implications of \eqref{eq:identifiable_dm} for estimation are profound. First, measuring the phase of the GHZ state gives a direct estimator for $\theta_0$ for all possible values. Whenever the phase flip is measured, a flip in the first channel occurred. Second, since the bit string describing both GHZ projectors in \eqref{eq:interghz} is $0$, the bits obtained by measuring the final mixed state characterizing the GHZ determines whether or not a bit-flip occurred in channel $j$. In particular, assume that the state in \eqref{eq:identifiable_dm} is measured in the GHZ basis, yielding the state $\ket*{\Phi_{s}^{b}}$. A bit-flip occurred in channel $0$ if, and only if, $b = 1$, while a flip occurred in channel $j > 0$  if, and only if, $s_j = 1$. Thus, we estimate all the parameters in the network by computing the number of times $b = 1$ and $s_j = 1$ in the strings obtained from GHZ measurements in a given set of observations.

In order to transform \eqref{eq:interghz0} into \eqref{eq:interghz} it is necessary to modify the circuits applied by the root and the intermediate node. The root circuit is incremented by applying the single qubit gate $XHX$ on the qubit that remains in the root after the CNOT, leading to $\mathcal{C}_0 = [H \otimes I, \text{CNOT}, XHX \otimes I]$. For the intermediate node, the circuit is extended with the application of the single qubit gate $HZ$ to the received qubit before using it as the control for the generalized $(n - 1)$-qubit Toffoli gate, which yields $\mathcal{C}_{n} = [HZ \otimes I^{\otimes n - 2}, T_{n - 1}]$.

The same circuit can be used to identify parameters for $Y$ channels with a modification on the estimators. When channels are described by $Y$, the intermediate node receives the state $\theta_0 \Phi_{0}^{0} + (1 - \theta_0)\Phi_{1}^{1}$. The estimators must change when such state is transmitted because the phase bit $b$ and the string $s$ determine together the occurrence of flips is the channel. It is possible to verify that a flip occurred in the first channel iff $b \neq \bigoplus_{j = 0}^{n - 1} s_j$. Since the occurrence of a flip in $\ch_0$ can always be detected, it is simple to relate $s_j$ to the occurrence of a flip in channel $j$. For the $Z$ case, it suffices to add the $(n - 1)$-qubit Hadamard gate $H^{\otimes n - 1}$ to the intermediate node circuit, such that $\mathcal{C}_0 = [H \otimes I, \text{CNOT}, HX \otimes I, H^{\otimes n - 1}]$. In addition, the $n - 1$ end-nodes receiving the qubits from the intermediate node must apply a Hardarmd gate before measuring in the GHZ basis. In this case, the bits characterizing the measured GHZ state provide direct estimators for the channel parameters as for $X$ channels.

\begin{figure}
    \begin{centering}
    \includegraphics[scale=0.46]{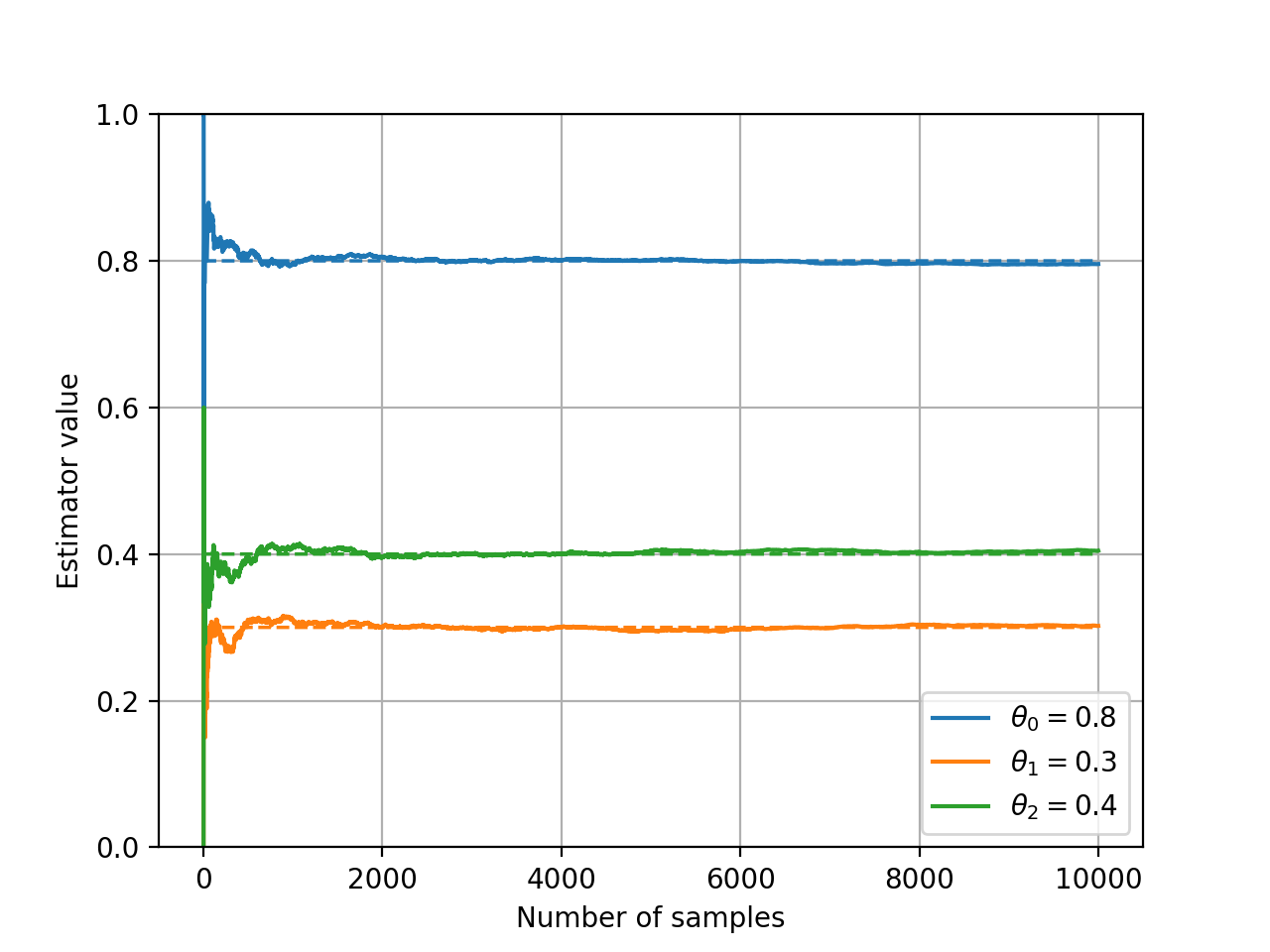}
    \end{centering}
    \caption{Esimators based on GHZ measurements for the four-node star graph with ground-truth parameter vector $\theta = [0.8, 0.3, 0.4]$. The estimators identify $\theta$ and provide a single solution for the tomography problem. Again, the dotted lines show ground-truth values. 
    }
    \label{fig:ghz}
\end{figure}

In order to compare the performance of the estimators based on GHZ measurement outcomes with the ones based on $Z$ basis measurements, we simulate the same four-node system reported in Figure \ref{fig:z} using the GHZ scheme. The results reported in Figure \ref{fig:ghz} show that the correct value for $\theta$ is identified and a single solution is obtained for the tomography problem. Moreover, the curves in Figure \ref{fig:ghz} are smoother than the ones in Figure \ref{fig:z}, what indicates that the variance for the GHZ-based estimator is less than the one for the $Z$-based estimator.

\subsection{Estimators and the QCRB}

The form of \eqref{eq:dmstrings} and \eqref{eq:identifiable_dm} fits into the definition of Theorem \ref{th:slddiag}, such that the QFIM follows \eqref{eq:QFIM} for both estimators. Moreover, we do not explicitly compute QFIM for the estimators, although it follows from Theorem \ref{th:slddiag} that we attain the QCRB in both cases because we use projective measurements on the basis that diagonalize the SLD of all parameters. Finally, the eigenvalues of $\rho(\theta)$ are first-order multivariate polynomials on both scenarios and evaluating \eqref{eq:QFIM} for such functions is straightforward, albeit space consuming.
\section{Conclusion}\label{sec:conclusion}

The quantum network tomography problem defined in this work connects quantum tomography with classical network tomography and targets the characterization of channels in a quantum network. The problem extends the quantum channel tomography problem to a network scenario considering that intermediate nodes do not provide information for estimation. Furthermore, we addressed the tomography problem in the context of networks with tree topology. We described a process for state distribution across trees that provides the necessary mixed states for end-nodes to measure. Our results for trees generalize to networks with arbitrary topology since graphs can be decomposed into trees and our methods can be applied to each tree in a decomposition. The estimators given for the star network with Pauli channels indicate that entanglement may provide advantages for tomography. The estimator obtained from global measurements outcomes in the end-nodes identifies the parameters without the need for any additional assumptions, in contrast with the lack of identifiability observed for local measurements. Such evidence motivates the description for the conditions under which entanglement enhances network tomography, as it has been previously investigated for other quantum estimation problems such as quantum sensing networks~\cite{proctor2018multiparameter}.

We identify three clear directions for future work. The first refers to the investigation of the star system for channels described by a single parameter and one generic unitary operator, as well as the case of the depolarizing channel. As a second direction, the description of the estimator that maximizes the QFIM for the star with single Pauli channels is of interest. Finally, framing the optimization problem to characterize the optimal way to partition a network into trees for tomography can bring insights on how our methods generalize to arbitrary networks.

{\em Acknowledgments}---This research was supported in part by the NSF grant CNS-1955834, NSF-ERC Center for Quantum Networks grant EEC-1941583 and by the MURI ARO Grant W911NF2110325.

\newpage
\bibliographystyle{unsrt}
\bibliography{references}


\end{document}